\documentclass[11pt]{article}
\usepackage{amsmath,amsfonts, amsthm, amssymb}
\usepackage{verbatim}
\usepackage[usenames,dvipsnames]{color}

\usepackage[margin=1in]{geometry}
\usepackage{hyperref}
\usepackage{palatino}
\usepackage{todonotes}
\renewcommand{\epsilon}{\varepsilon}
\providecommand{\eps}{\epsilon}
\renewcommand{\le}{\leqslant}
\renewcommand{\leq}{\leqslant}

\renewcommand{\geq}{\geqslant}

\newtheorem{thm}{Theorem}[section]\newtheorem*{rethm}{Theorem}
    \newtheorem{lemma}[thm]{Lemma}
    \newtheorem{cor}[thm]{Corollary}\newtheorem{prop}[thm]{Proposition}
    
\theoremstyle{definition}
    \newtheorem{defn}[thm]{Definition}
\theoremstyle{remark}\newtheorem*{rmk}{Remark}

\providecommand{\F}{\mathbb{F}}

\providecommand{\abs}[1]{\lvert#1\rvert}

\DeclareMathOperator{\poly}{poly}

\title{Deletion codes in the high-noise and high-rate regimes}
\author{Venkatesan Guruswami\thanks{Some of this work was done when
    the author was a visiting researcher at Microsoft Research New
    England. Email: {\tt guruswami@cmu.edu}. Research supported in
    part by NSF grants CCF-0963975 and CCF-1422045.} \and Carol
  Wang\thanks{Part of this work was done on a visit to Microsoft
    Research New England. Email: {\tt wangc@cs.cmu.edu}. Supported by
    an NSF Graduate Research Fellowship and NSF CCF-0963975.}}

\date{Computer Science Department \\ Carnegie Mellon University \\ Pittsburgh, PA}

\begin{document}
\maketitle
\thispagestyle{empty}

\begin{abstract}
The noise model of deletions poses significant challenges in coding theory, with basic questions like the capacity of the binary deletion channel still being open. In this paper, we study the harder model of {\em worst-case}
deletions, with a focus on constructing efficiently decodable codes for the two extreme regimes of high-noise and high-rate. Specifically, we construct polynomial-time decodable codes with the following trade-offs (for any $\eps > 0$):
\begin{enumerate}
\item[(i)] Codes that can correct a fraction $1-\eps$ of deletions with rate $\poly(\eps)$ over an alphabet of size $\poly(1/\eps)$;
\item[(ii)] Binary codes of rate $1-\tilde{O}(\sqrt{\eps})$ that can correct a fraction $\eps$ of deletions; and
\item[(iii)] Binary codes that can be {\em list decoded} from a fraction $(1/2-\eps)$ of deletions with rate $\poly(\eps)$.
\end{enumerate}

\smallskip
Our work is the first to achieve the qualitative goals of correcting a
deletion fraction approaching $1$ over bounded alphabets, and
correcting a constant fraction of bit deletions with rate approaching
$1$ over a fixed alphabet. The above results bring our understanding
of deletion code constructions in these regimes to a similar level as
worst-case errors.
\end{abstract}

\newpage
\section{Introduction}

This work addresses the problem of constructing codes which can be efficiently corrected from a constant fraction of worst-case deletions. In contrast to erasures, the locations of deleted symbols are {\em not} known to the decoder, who receives only a subsequence of the original codeword. The deletions can be thought of as corresponding to errors in synchronization during communication. The loss of position information makes deletions a very challenging model to cope with, and our understanding of the power and limitations of codes in this model significantly lags behind what is known for worst-case errors.

The problem of communicating over the {\it binary deletion channel}, in which each transmitted bit is deleted 
independently with a fixed probability $p$, has been a subject of much study (see the excellent survey by Mitzenmacher~\cite{m-survey} for more background and references). However, even this easier case is not well-understood. In particular, the capacity of the binary deletion channel remains open, although it is known to approach $1-h(p)$ as $p$ goes to $0$, where $h(p)$ is the binary entropy function (see~\cite{DG,gallager,zigangirov} for lower bounds and~\cite{KMS, KM} for upper bounds), and it is known to be positive (at least $(1-p)/9$)~\cite{MD06}) even as $p \to 1$. 

The more difficult problem of correcting from adversarial rather than random deletions has also been studied, but with a focus on correcting a constant {\em number} (rather than fraction) of deletions. It turns out that obtaining optimal trade-offs  to correct a single deletion is already a non-trivial and rich problem (see~\cite{sloane-deletion}), and we do not yet have a good understanding for two or more deletions. 

Coding for a constant {\em fraction} of adversarial deletions has been considered previously by Schulman and Zuckerman~\cite{SZ}. They construct constant-rate binary codes which are efficiently decodable from a small constant fraction of worst-case deletions and insertions, and can also handle a small fraction of transpositions. The rate of these codes are bounded away from $1$, whereas existentially one can hope to achieve a rate approaching $1$ for a small deletion fraction.

The central theoretical goal in error-correction against any specific noise model is to understand the combinatorial trade-off between the rate of the code and noise rate that can be corrected, and to construct codes with efficient error-correction algorithms that ideally approach this optimal trade-off. While this challenge is open in general even for the well-studied and simpler model of errors and erasures, in the case of worst-case deletions, our knowledge has even larger gaps. (For instance, we do not know the largest deletion fraction which can be corrected with positive rate for any fixed alphabet size.)
Over large alphabets that can grow with the length of the code, we can include the position of each codeword symbol as a header that is part of the symbol. This reduces the model of deletions to that of erasures, where simple optimal constructions (eg. Reed-Solomon codes) are known. 

Given that we are far from an understanding of the best rate
achievable for any specified deletion fraction, in this work we focus
on the two extreme regimes --- when the deletion fraction is small
(and the code rate can be high), and when the deletion fraction
approaches the maximum tolerable value (and the code rate is
small). Our emphasis is on constructing codes that can be efficiently
encoded and decoded, with trade-offs not much worse than
random/inefficient codes (whose parameters we compute in
Section~\ref{sec:exis}).  Our results, described next, bring the level
of knowledge on efficient deletion codes in these regimes to a roughly
similar level as worst-case errors. There are numerous open questions,
both combinatorial and algorithmic, that remain open, and it is our
hope that the systematic study of codes for worst-case deletions
undertaken in this work will spur further research on good
constructions beyond the extremes of low-noise and high-noise.

\subsection{Our results}

The best achievable rate against a fraction $p$ of deletions cannot exceed $1-p$, as we need to be able to recover the message from the first $(1-p)$ fraction of codeword symbols. As mentioned above, over large (growing) alphabets this trade-off can in fact be achieved by a simple reduction to the model of erasures. Existentially, as we show in Section~\ref{sec:exis}, for any $\gamma > 0$, it is easy to show that there are codes of rate $1-p-\gamma$ to correct a fraction $p$ of deletions over an alphabet size that depends only on $\gamma$. For the weaker model of erasures, where the receiver knows the locations of erased symbols, we know explicit codes,  namely certain algebraic-geometric codes~\cite{shum-etal} or expander based constructions~\cite{AEL,GI-ieeeit}, achieving the optimal trade-off (rate $1-p-\gamma$ to correct a fraction $p$ of erasures) over alphabets growing only as a function of $1/\gamma$. For deletions, we do not know how to construct codes with such a trade-off efficiently. However, in the high-noise regime when the deletion fraction is $p=1-\eps$ for some small $\eps > 0$, we are able to construct codes of rate $\poly(\eps)$
over an alphabet of size $\poly(1/\eps)$. Note that an alphabet of size at least $1/\eps$ is needed, and the rate can be at most $\eps$, even for the simpler model of erasures, so we are off only by polynomial factors.

\begin{rethm}[Theorem~\ref{thm:const-adv}] Let $1/2>\epsilon>0$. There is an explicit code of rate $\Omega(\epsilon^2)$ 
and alphabet size $\poly(1/\epsilon)$ 
which can be corrected from a $1-\epsilon$ fraction of worst-case deletions. 

Moreover, this code can be constructed, encoded, and decoded in time $N^{\poly(1/\epsilon)}$, where $N$ is the block length of the code. 
\end{rethm}

The above handles the case of very large fraction of deletions. At the other extreme, when the deletion fraction is small, the following result shows that we  achieve high rate (approaching one) even over the binary alphabet. 

\begin{rethm}[Theorem~\ref{thm:low-del}] Let $\epsilon>0$. There is an explicit binary code $C\subseteq\{0,1\}^N$ which is decodable from an $\epsilon$ fraction of deletions with rate $1-\tilde{O}(\sqrt{\epsilon})$ 
in time $N^{\poly(1/\epsilon)}$. 

Moreover, $C$ can be constructed and encoded in time $N^{\poly(1/\epsilon)}$. 
\end{rethm}

\medskip
The next question is motivated by constructing binary codes for the
``high noise" regime. In this case, we do not know (even
non-constructively) the minimum fraction of deletions that forces the rate of the code to approach zero. (Contrast this with the situation for erasures
(resp. errors), where we know the zero-rate threshold to be an erasure fraction
$1/2$ (resp. error fraction $1/4$).) Clearly, if the adversary can
delete half of the bits, he can always ensure that the decoder
receives $0^{n/2}$ or $1^{n/2}$, so at most two strings can be
communicated. Surprisingly, in the model of {\em list decoding}, where
the decoder is allowed to output a small list consisting of all
codewords which contain the received string as a subsequence, one can
in fact decode from an deletion fraction arbitrarily close to $1/2$,
as our third construction shows:

\begin{rethm}[Theorem~\ref{thm:list}] Let $0<\epsilon<1/2$. There is an explicit binary code $C\subseteq\{0,1\}^N$ of rate $\tilde{\Omega}(\epsilon^3)$ which is list-decodable from a $1/2-\epsilon$ fraction of deletions with list size $(1/\eps)^{O(\log\log(1/\eps))}$. 

This code can be constructed, encoded, and list-decoded in time $N^{\poly(1/\epsilon)}$. 
\end{rethm}

We should note that it is not known if list decoding is required to correct deletion fractions close to $1/2$, or if one can get by with unique decoding. Our guess would be that the largest deletion fraction unique decodable with binary codes is (noticeably) bounded away from $1/2$. The cubic dependence on $\eps$ in the rate in the above theorem is similar to what has been achieved for correcting $1/2-\eps$ fraction of errors~\cite{GR-FRS}. We anticipate (but have not formally checked) that a similar result holds over any fixed alphabet size $k$ for list decoding from a fraction $(1-1/k-\eps)$ of symbol deletions.

\medskip\noindent {\bf Construction approach.}  Our codes, like many
considered in the past, including those of~\cite{CMNN,DM,Ratzer} in
the random setting and particularly~\cite{SZ} in the adversarial
setting, are based on concatenating a good error-correcting code (in
our case, Reed-Solomon or Parvaresh-Vardy codes) with an inner
deletion code over a much smaller block length. This smaller block
length allows us to find and decode the inner code using brute
force. The core of the analysis lies in showing that the adversary can
only affect the decoding of a bounded fraction of blocks of the inner
code, allowing the outer code to decode using the remaining blocks.

While our proofs only rely on elementary combinatorial arguments, some
care is needed to execute them without losing in rate (in the case of
Theorem~\ref{thm:low-del}) or in the deletion fraction we can handle
(in the case of Theorems \ref{thm:const-adv} and \ref{thm:list}).  In
particular, for handling close to fraction $1$ of deletions, we have
to carefully account for errors and erasures of outer Reed-Solomon
symbols caused by the inner decoder. To get binary codes of rate
approaching $1$, we separate inner codeword blocks with (not too long)
buffers of $0$'s and we exploit some additional structural properties
of inner codewords that necessitate many deletions to make them
resemble buffers. The difficulty in both these results is unique
identification of enough inner codeword boundaries so that the
Reed-Solomon decoder will find the correct message. The list decoding
result is easier to establish, as we can try many different boundaries
and use a ``list recovery" algorithm for the outer algebraic code. To
optimize the rate, we use the Parvaresh-Vardy codes~\cite{PV-focs05}
as the outer algebraic code.

\subsection{Organization}

In Section~\ref{sec:exis}, we consider the performance of certain random and greedily constructed codes. These serve both as benchmarks and as starting points for our efficient constructions. In Section~\ref{sec:const-adv}, we construct codes in the high deletion regime. In Section~\ref{sec:low-del}, we give high-rate binary codes which can correct a small constant fraction of deletions. In Section~\ref{sec:list}, we give list-decodable binary codes up to the optimal error fraction. Some open problems appear in Section~\ref{sec:open}. Omitted proofs appear in the appendices.

\section{Existential bounds for deletion codes}
\label{sec:exis}

 A quick recap
of standard coding terminology: a code $C$ of block length $n$ over an
alphabet $\Sigma$ is a subset $C \subseteq \Sigma^n$. The rate of $C$
is defined as $\frac{\log |C|}{n \log |\Sigma|}$. The encoding
function of a code is a map $E : [|C|] \to \Sigma^n$ whose image
equals $C$ (with messages identified with $[|C|]$ in some canonical
way). Our constructions all exploit the simple but powerful idea of
code concatenation: If $C_{\mathrm{out}} \subseteq
\Sigma_{\mathrm{out}}^n$ is an ``outer'' code with encoding function
$E_{\mathrm{out}}$, and $C_{\mathrm{in}} \subseteq
\Sigma_{\mathrm{in}}^m$ is an ``inner'' code encoding function
$E_{\mathrm{in}}: \Sigma_{\mathrm{out}} \to \Sigma_{\mathrm{in}}^m$,
the the concatenated code $C_{\mathrm{out}} \circ C_{\mathrm{in}}
\subseteq \Sigma_{\mathrm{in}}^{nm}$ is a code whose encoding function
first applies $E_{\mathrm{out}}$ to the message, and then applies
$E_{\mathrm{in}}$ to each symbol of the resulting outer codeword.

In this section, we show the existence of deletion codes in certain ranges of parameters, without the requirement of efficient encoding or decoding. The proofs (found in the appendix) follow from standard probabilistic arguments, but to the best of our knowledge, these bounds were not known previously. The codes of Theorem~\ref{low-del} will be used as inner codes in our final concatenated constructions. 

Throughout, we will write $[k]$ for the set $\{1,\dotsc, k\}$. 
We will also use the binary entropy function, defined for $\delta\in[0,1]$ as $h(\delta) = \delta \log\frac{1}{\delta} + (1-\delta) \log\frac{1}{1-\delta}$. All logarithms in the paper are to base $2$.

\medskip

We note that constructing a large code over $[k]^m$ which can correct from a $\delta$ fraction of deletions is equivalent to constructing a large set of strings such that for each pair, their longest common subsequence (LCS) has length less than $(1-\delta)m$. 
\medskip

We first consider how well a random code performs, using the following theorem from \cite{KLM}, which upper bounds the probability that a pair of randomly chosen strings has a long LCS. 

\begin{thm}[\cite{KLM}, Theorem 1] For every $\gamma>0$, there exists $c>0$ such that if $k$ and $m/\sqrt{k}$ are sufficiently large, and $u,v$ are chosen independently and uniformly from $[k]^m$, then 
\[\Pr\left[\bigl|\mathrm{LCS}(u,v) - 2m/\sqrt{k}\bigr|\geq \frac{\gamma m}{\sqrt{k}}\right]\leq e^{-cm/\sqrt{k}}.  
\]
\end{thm}

Fixing $\gamma$ to be $1$, we obtain the following. 

\begin{prop} Let $\epsilon>0$ be sufficiently small and let $k=(4/\epsilon)^2$. There exists a code $C\subseteq[k]^m$ of rate $R=O\bigl(\epsilon/\log(1/\epsilon)\bigr)$ which can correct a $1-\epsilon=1-4/\sqrt{k}$ fraction of deletions. 
\end{prop}

The following results, and in particular Corollary~\ref{high-del}, show that we can nearly match the performance of random codes using a simple greedy algorithm. 

We first bound the number of strings which can have a fixed string $s$ as a subsequence. 

\begin{lemma}\label{lem:count} Let $\delta\in(0,1/k)$, set $\ell=(1-\delta) m$, and let $s\in[k]^{\ell}$. The number of strings $s'\in[k]^m$ containing $s$ as a subsequence is at most 
\[\sum_{t=\ell}^m \binom{t-1}{\ell-1} k^{m-t}(k-1)^{t-\ell}\leq k^{m-\ell} \binom{m}{\ell}.\]

When $k=2$, we have the estimate 
\[\sum_{t=\ell}^m \binom{t-1}{\ell-1} 2^{m-t}\leq \delta m \binom{m}{\ell}.\]
\end{lemma}

\begin{thm} \label{low-del} Let $\delta,\gamma>0$. Then for every $m$, there exists a code $C\subseteq[k]^m$ of rate $R=1-\delta-\gamma$ such that:
\begin{itemize}
\item $C$ can be corrected from a $\delta$ fraction of worst-case deletions, 
provided $k\geq 2^{2h(\delta)/\gamma}$. 
\item $C$ can be found, encoded, and decoded in time $k^{O(m)}$. 
\end{itemize}

Moreover, when $k=2$, we may take $R=1-2h(\delta)-\log(\delta m)/m$. 
\end{thm}

\begin{rmk} The authors of~\cite{KMTU} show a similar result for the binary case, but use the weaker bound in Lemma~\ref{lem:count} to get a rate of $1-\delta - 2h(\delta)$. 
\end{rmk}
\medskip

With a slight modification to the proof of Theorem~\ref{low-del}, we obtain the following construction, which will be used in Section~\ref{sec:low-del}. 

\begin{prop}\label{thm:balanced} Let $\delta,\beta\in(0,1)$. Then for every $m$, there exists a code $C\subseteq\{0,1\}^m$ of rate $R=1-2h(\delta) - O(\log (\delta m)/m) - 2^{-\Omega(\beta m)}/m$ such that:
\begin{itemize}
\item For every string $s\in C$, $s$ is ``$\beta$-dense'': every interval of length $\beta m$ in $s$ contains at least $\beta m/10 $ $1$'s, 
\item $C$ can be corrected from a $\delta$ fraction of worst-case deletions, and
\item $C$ can be found, encoded, and decoded in time $2^{O(m)}$. 
\end{itemize}
\end{prop}

In the high-deletion regime, we have the following corollary to Theorem~\ref{low-del}, obtained by setting $\delta=1-\epsilon$ and $\gamma = (1-\theta)\epsilon$, and noting that $h(\epsilon)\leq \epsilon \log(1/\epsilon) + 2\epsilon$ when $\epsilon<1/2$. 

\begin{cor} \label{high-del} Let $1/2>\epsilon>0$ and $\theta\in (0,1/3]$. There for every $m$, exists a code $C\subseteq[k]^m$ of rate $R=\epsilon\cdot \theta$ which can correct a $1-\epsilon$ fraction of deletions in time $k^{O(m)}$, provided $k\geq 64/\epsilon^{\frac{2}{1-\theta}}$. 
\end{cor}

\section{Coding against $1-\epsilon$ deletions}
\label{sec:const-adv}

In this section, we construct codes for the high-deletion regime. More precisely, we have the following theorem. 

\begin{thm} \label{thm:const-adv} Let $1/2>\epsilon>0$. There is an explicit code of rate $\Omega(\epsilon^2)$ 
and alphabet size $\poly(1/\epsilon)$ 
which can be corrected from a $1-\epsilon$ fraction of worst-case deletions. 

Moreover, this code can be constructed, encoded, and decoded in time $N^{\poly(1/\epsilon)}$, where $N$ is the block length of the code. 
\end{thm}
\medskip

We first define the code. Theorem~\ref{thm:const-adv} is then a direct corollary of Lemmas~\ref{thm:high-del-rate} and~\ref{thm:high-del-alg}. 
\medskip

\noindent {\bf The code}: Our code will be over the alphabet $\{0,1,\dots,D-1\} \times [k]$, where $D=8/\epsilon$ and $k=O(1/\epsilon^3)$. 
\smallskip

We first define a code $C'$ over the alphabet $[k]$ by concatenating a Reed-Solomon code with a deletion code constructed using Corollary~\ref{high-del}, setting $\theta=1/3$.  

More specifically, let $\F_q$ be a finite field. For any $n'\leq n\leq q$, the Reed-Solomon code of length $n\leq q$ and dimension $n'$ is a subset of $\F_q^n$ which admits an efficient algorithm to uniquely decode from $t$ errors and $r$ erasures, provided $r+2t<n-n'$ (see, for example,~\cite{WB}). 
\smallskip

In our construction, we will take
$n=q=2n'/\epsilon$. We first encode our message to a codeword
$c=(c_1,\dotsc, c_n)$ of the Reed-Solomon code. For each $i$, we then
encode the pair $(i, c_i)$ using Corollary~\ref{high-del} by a code
$C_1\colon [n]\times\F_q\to [k]^m$, where $m=12\log q/\epsilon$, which
can correct a $1-\epsilon/2$ fraction of deletions.

\medskip

To obtain our final code $C$, we replace every symbol $s$ in $C'$ which encodes the $i$th RS coordinate by the pair $\bigl(i\pmod{D}, s\bigr)\in \{0,1,\dots,D-1\} \times[k]$. The first coordinate, $i\pmod{D}$, contains the location of the codeword symbol modulo $D$, and we will refer to it as a {\bf header}.

\begin{lemma} 
\label{thm:high-del-rate} 
The rate of $C$ is $\Omega(\epsilon^2)$. 
\end{lemma}

\begin{proof}  
The rate of the outer Reed-Solomon code, labeled with indices, is at
least $\epsilon/4$. The rate of the inner code can be taken to be
$\Omega(\epsilon)$, by Corollary~\ref{high-del}. Finally, the alphabet
increase in transforming $C'$ to $C$ decreases the rate by a factor of
$\frac{\log(k)}{\log(D k)}=\Omega(1)$.

In particular, this gives us a final rate of $\Omega(\epsilon^2)$.
\end{proof}

\begin{lemma} \label{thm:high-del-alg} The code $C$ can be decoded from a $1-\epsilon$ fraction of worst-case deletions in time $N^{O(\poly1/\epsilon)}$. 
\end{lemma}

\begin{proof}

Let $N$ be the block length of $C$. We apply the following algorithm to decode $C$. 

\begin{itemize}
\item[-] We partition the received word into {\em blocks} as follows: The first block begins at the first coordinate, and each subsequent block begins at the next coordinate whose header differs from its predecessor. This takes time $\poly(N)$. 

\item[-] We begin with an empty set $L$. 

For each block which is of length between $\epsilon m/2$ and $m$, we remove the headers by replacing each symbol $(a,b)$ with the second coordinate $b$. We then apply the decoder from Corollary~\ref{high-del} to the block. If this succeeds, outputting a pair $(i, r_i)$, then we add $(i, r_i)$ to $L$. This takes time $N^{\poly(1/\epsilon)}$. 

\item[-] If for any $i$, $L$ contains multiple pairs with first coordinate $i$, we remove all such pairs from $L$. $L$ thus contains at most one pair $(i, r_i)$ for each index $i$. We apply the Reed-Solomon decoding algorithm to the string $r$ whose $i$th coordinate is $r_i$ if $(i, r_i)\in L$ and erased otherwise. This takes time $\poly(N)$. 
\end{itemize}
\medskip

{\bf Analysis}: For any $i$, we will decode a correct coordinate $\bigl(i, c_i\bigr)$ if there is a block of length at least $\epsilon m/2$ which is a subsequence of $C_1(i, c_i)$. (Here and in what follows we abuse notation by disregarding headers on codeword symbols.)

Thus, the Reed-Solomon decoder will receive the correct value of the $i$th coordinate unless one of the following occurs: 
\begin{enumerate}
\item (Erasure) The adversary deletes a $\geq 1-\epsilon/2$ fraction of $C_1(i, c_i)$. 
\item (Merge) The block containing (part of) $C_1(i, c_i)$ also contains symbols from other codewords of $C_1$, because the adversary has erased the codewords separating $C_1(i, c_i)$ from its neighbors with the same header. 
\item (Conflict) Another block decodes to $(i, r)$ for some $r$. Note that an erasure cannot cause a coordinate to decode incorrectly, so a conflict can only occur from a merge. 
\end{enumerate}
\smallskip

We would now like to bound the number of errors and erasures the adversary can cause. 

\begin{itemize}
\item[-] If the adversary causes an erasure without causing a merge, this requires at least $(1-\epsilon/2) m$ deletions within the block which is erased, and no other block is affected. 

\item[-] If the adversary merges $t$ inner codewords with the same label, this requires at least $(t-1)(D-1)m$ deletions, of the intervening codewords with different labels. The merge causes the fully deleted inner codewords to be erased, and causes the $t$ merged codewords to resolve into at most one (possibly incorrect) value. This value, if incorrect, could also cause one conflict. 

In summary, in order to cause one error and $r\leq (t-1)D +2$ erasures, the adversary must introduce at least $(t-1)(D-1)m\geq (2+r) (1-\epsilon/2)m$ deletions. 
\end{itemize}

In particular, if the adversary causes $s$ errors and $r_1$ erasures by merging, and $r_2$ erasures without merging, this requires at least 
\[\geq (2s + r_1) (1-\epsilon/2)m + r_2(1-\epsilon/2)m = (2s+r)(1-\epsilon/2)m\]
deletions. Thus, when the adversary deletes at most a $(1-\epsilon)$ fraction of codeword symbols, we have that $2s+r$ is at most $(1-\epsilon)mn/(1-\epsilon/2)m<n(1-\epsilon/2)$. Recalling that the Reed-Solomon decoder in the final step will succeed as long as $2s+r<n(1-\epsilon/2)$, we conclude that the decoding algorithm will output the correct message. 
\end{proof}

\section{Binary codes against $\epsilon$ deletions}
\label{sec:low-del}

\subsection{Construction overview}

The goal in our constructions is to allow the decoder to approximately locate the boundaries between codewords of the inner code, in order to recover the symbols of the outer code. In the previous section, we were able to achieve this by augmenting the alphabet and letting each symbol encode some information about the block to which it belongs. In the binary case, we no longer have this luxury. 

The basic idea of our code is to insert long runs of zeros, or ``buffers,'' between adjacent inner codewords. The buffers are long enough that the adversary cannot destroy many of them. If we then choose the inner code to be dense (in the sense of Proposition~\ref{thm:balanced}), it is also difficult for a long interval in any codeword to be confused for a buffer. This approach optimizes that of~\cite{SZ}, which uses an inner code of rate $1/2$ and thus has final rate bounded away from $1$. 

The balance of buffer length and inner codeword density seems to make buffered codes unsuited for high deletion fractions, and indeed our results only hold as the deletion fraction goes to zero. 

\subsection{Our construction}

We now give the details of our construction. For simplicity, we will not optimize constants in the analysis. 

\begin{thm}\label{thm:low-del} Let $\epsilon>0$. There is an explicit binary code $C\subseteq\{0,1\}^N$ which is decodable from an $\epsilon$ fraction of deletions with rate $1-\tilde{O}(\sqrt{\epsilon})$ 
in time $N^{\poly(1/\epsilon)}$. 

Moreover, $C$ can be constructed and encoded in time $N^{\poly(1/\epsilon)}$. 
\end{thm}
\medskip

{\bf The code}: 
We again use a concatenated construction with a Reed-Solomon code as the outer code, choosing one which can correct a $12\sqrt{\epsilon}$ fraction of errors and erasures. For each $i$, we replace the $i$th coordinate $c_i$ with the pair $(i, c_i)$. In order to ensure that the rate stays high, we use a RS code over $\F_{q^h}$, with block length $n=q$, where we will take $h=1/\epsilon$. 
\medskip

The inner code will be a good binary deletion code $C_1$ of block length $m$ correcting a $\delta=40\sqrt{\epsilon}$ fraction of deletions, found using Proposition~\ref{thm:balanced}, with $\beta=\delta/4$. Recall that this code only contains ``$\beta$-dense strings,'' for which any interval of length $\beta m$ contains $\beta m/10$ $1$'s. We will assume each codeword begins and ends with a $1$. 

Now, between each pair of adjacent inner codewords of $C_1$, we insert a {\em buffer} of $\delta m$ zeros. This gives us our final code $C$. 
\medskip

\begin{lemma} The rate of $C$ is $1-\tilde{O}(\sqrt{\epsilon})$. 
\end{lemma}

\begin{proof} The rate of the outer (labeled) Reed-Solomon code is $(1-24\sqrt{\epsilon})\cdot \frac{h}{h+1}$. The rate of the inner code $C_1$ can be taken to be $1-2h(\delta)-o(1)$, by Proposition~\ref{thm:balanced}. Finally, adding buffers reduces the rate by a factor of $\frac{1}{1+\delta}$. 

Combining these with our choice of $\delta$, we get that the rate of $C$ is $1-\tilde{O}(\sqrt{\epsilon})$. 
\end{proof}

\begin{lemma} The code $C$ can be decoded from an $\epsilon$ fraction of worst-case deletions in time $N^{\poly(1/\epsilon)}$. 
\end{lemma}

{\bf The algorithm}: 
\begin{itemize}
\item[-] The decoder first locates all runs of at least $\delta m/2$ contiguous zeroes in the received word. These runs (``buffers'') are removed, dividing the codeword into blocks of contiguous symbols which we will call {\it decoding windows}. Any leading zeroes of the first decoding window and trailing zeroes of the last decoding window are also removed. 
This takes time $\poly(N)$. 

\item[-] We begin with an empty set $L$. 

For each decoding window, we apply the decoder from Proposition~\ref{thm:balanced} to attempt to recover a pair $(i, r_i)$. If we succeed, this pair is added to $L$. This takes time $N^{\poly(1/\epsilon)}$. 

\item[-] If for any $i$, $L$ contains multiple pairs with first coordinate $i$, we remove all such pairs from $L$. $L$ thus contains at most one pair $(i, r_i)$ for each index $i$. We apply the Reed-Solomon decoding algorithm to the string $r$ whose $i$th coordinate is $r_i$ if $(i, r_i)\in L$ and erased otherwise, attempting to recover from a $12\sqrt{\epsilon}$ fraction of errors and erasures. This takes time $\poly(N)$. 
\end{itemize}
\medskip

{\bf Analysis}: Notice that if no deletions occur, the decoding windows will all be codewords of the inner code $C_1$, which will be correctly decoded. At a high level, we will show that the adversary cannot corrupt many of these decoding windows, even with an $\epsilon$ fraction of deletions. 
\medskip

We first show that the number of decoding windows considered by our algorithm is close to $n$, the number of windows if there are no deletions. 
\begin{lemma} \label{lem:windows} If an $\epsilon$ fraction of deletions have occurred, then the number of decoding windows considered by our algorithm is between $(1-2\sqrt{\epsilon}) n$ and $(1 + 2\sqrt{\epsilon})n$. 
\end{lemma}
\begin{proof} Recall that the adversary can cause at most $\epsilon n m(1 + \delta)\leq 2\epsilon nm$ deletions. 
\medskip

Upper bound: The adversary can increase the number of decoding windows only by creating new runs of $\delta m/2$ zeroes (that are not contained within a buffer). Such a new run must be contained entirely within an inner codeword $w\in C_1$. However, as $w$ is $\delta/4$-dense, in order to create a run of zeroes of length $\delta m/2$, at least $\delta m/20=2\sqrt{\epsilon}$ $1$'s must be deleted for each such run. In particular, at most $\sqrt{\epsilon} n$ blocks can be added. 
\smallskip

Lower bound: The adversary can decrease the number of decoding windows only by decreasing the number of buffers. He can achieve this either by removing a buffer, or by merging two buffers. Removing a buffer requires deleting $\delta m/2=20\sqrt{\epsilon}m$ zeroes from the original buffer. Merging two buffers requires deleting all $1$'s in the inner codewords between them. As inner codewords are $\delta/4$-dense, this requires at least $\sqrt{\epsilon}m$ deletions for each merged buffer. In particular, at most $2 \sqrt{\epsilon} n$ buffers can be removed. 
\end{proof}

We now show that almost all of the decoding windows being considered are decoded correctly by the inner decoder. 
\begin{lemma} The number of decoding windows which are incorrectly decoded is at most $4\sqrt{\epsilon}n$. 
\end{lemma}
\begin{proof} The inner decoder will succeed on each decoding window which is a subsequence of a valid inner codeword $w\in C_1$ of length at least $(1-\delta)m$. This will happen unless: 

\begin{enumerate}
\item The window is too short:
\begin{itemize}
\item[(a)] a subsequence of $w$ has been marked as a (new) buffer, or
\item[(b)] a $\rho$ fraction of $w$ has been marked as part of the adjacent buffers, combined with a $\delta-\rho$ fraction of deletions within $w$. 
\end{itemize}
\item The window is not a subsequence of a valid inner codeword: the window contains buffer symbols and/or a subsequence of multiple inner codewords.
\end{enumerate}

We first show that (1) holds for at most $3\sqrt{\epsilon} n$ windows. 

From the proof of Lemma~\ref{lem:windows}, there can be at most $\sqrt{\epsilon} n$ new buffers introduced, thus handling Case 1(a). In Case 1(b), if $\rho<\delta/2$, then there must be $\delta/2$ deletions within $w$. On the other hand, if $\rho\geq \delta/2$, one of two buffers adjacent to $w$ must have absorbed at least $\delta m/4$ symbols of $w$, so as $w$ is $\delta/4$-dense, this requires $\delta m/40=\sqrt{\epsilon} m$ deletions, so can occur in at most $2\sqrt{\epsilon} n$ windows. 
\medskip

We also have that (2) holds for at most $\sqrt{\epsilon} n$ windows, as at least $\delta m/2$ symbols must be deleted from a buffer in order to prevent the algorithm from marking it as a buffer. As in Lemma~\ref{lem:windows}, this requires $20\sqrt{\epsilon}$ deletions for each merged window, and so there are at most $\sqrt{\epsilon n}$ windows satisfying case~(2). 
\end{proof}

We now have that the inner decoder outputs $(1-6\sqrt{\epsilon})n$ correct values. After removing possible conflicts in the last step of the algorithm, we have at least $(1-12\sqrt{\epsilon}) n$ correct values, so that the Reed-Solomon decoder will succeed and output the correct message. 

\section{List-decoding binary deletion codes}
\label{sec:list}

The results of Section~\ref{sec:low-del} show that we can have good explicit binary codes when the deletion fraction is low. In this section, we address the opposite regime, of high deletion fraction. 
As a first step, notice that in any reasonable model, including list-decoding, we cannot hope to efficiently decode from a $1/2$ deletion fraction with a polynomial list size and constant rate. With block length $n$ and $n/2$ deletions, the adversary can ensure that what is received is either $n/2$ $1$'s or $n/2$ $0$'s. 

Thus, for binary codes and $\epsilon>0$, we will consider the question of whether it is possible to list decode from a fraction $1/2-\epsilon$ of deletions. 
\medskip

\begin{defn}
We say that a code $C\subseteq\{0,1\}^m$ is list-decodable from a $\delta$ deletion fraction with list size $L$ if every sequence of length $(1-\delta) m$ is a subsequence of at most $L$ codewords. If this is the case, we will call $C$ $(\delta, L)$ {\it list-decodable from deletions}.
\end{defn}
\medskip

\begin{rmk} Although the results of this section are proven in the setting of list-decoding, it is {\em not} known that we cannot have unique decoding of binary codes up to deletion fraction $1/2-\epsilon$. See the first open problem in Section~\ref{sec:open}. 
\end{rmk}

\subsection{List-decodable binary deletion codes (existential)}
\label{sec:list-exis}

In this section, we show that good list-decodable codes exist. This construction will be the basis of our explicit construction of list-decodable binary codes.  The proof appears in the appendix. 

\begin{thm} \label{binary-list} Let $\delta,L>0$. Let $C\subseteq\{0,1\}^m$ consist of $2^{Rm}$ independently, uniformly chosen strings, where $R\leq 1 - h(\delta) - 3/L$. Then $C$ is $\bigl(\delta, L\bigr)$ list-decodable from deletions with probability at least $1-2^{-m}$. 

Moreover, such a code can be constructed and decoded in time $2^{\poly(mL)}$. 
\smallskip

In particular, when $\delta = 1/2-\epsilon$, we can construct and decode in time $2^{\poly(m/\epsilon)}$ a code $C\subseteq \{0,1\}^m$ of rate $\Omega(\epsilon^2)$ which is $\bigl(\delta, O(1/\epsilon^2)\bigr)$ list-decodable from deletions. 
\end{thm}

\subsection{List-decodable binary deletion codes (explicit)}

We now use the existential construction of Theorem~\ref{binary-list} to give an explicit construction of constant-rate list-decodable binary codes. Our code construction uses  Parvaresh-Vardy codes (\cite{PV-focs05}) as outer codes, and an inner code constructed using Section~\ref{sec:list-exis}. 

The idea is to list-decode ``enough'' windows and then apply the list recovery algorithm of Theorem~\ref{pv-list}. 

\begin{thm} \label{thm:list} Let $0<\epsilon<1/2$. There is an explicit binary code $C\subseteq\{0,1\}^N$ of rate $\tilde{\Omega}(\epsilon^3)$ which is list-decodable from a $1/2-\epsilon$ fraction of deletions with list size $(1/\epsilon)^{O(\log\log \epsilon)}$. 

This code can be constructed, encoded, and list-decoded in time $N^{\poly(1/\epsilon)}$. 
\end{thm}

We will appeal in our analysis to the following theorem, which can be
found in~\cite{GR-soft}.

\begin{thm}[\cite{GR-soft}, Corollary 5] \label{pv-list}
For all integers $s\geq 1$, for all prime powers $r$, every pair of integers $1<K\leq N\leq q$, there is an explicit $\F_r$-linear map $E\colon \F_q^K\to \F_{q^s}^N$ whose image $C'$ is a code satisfying:

\begin{itemize}
\item[-] There is an algorithm which, given a collection of subsets
  $S_i\subseteq \F_{q^s}$ for $i\in[N]$ with $\sum_i \abs{S_i}\leq N
  \ell$, runs in $\poly\bigl((rs)^s, q,\ell\bigr)$ time, and outputs a
  list of size at most $O\bigl((rs)^s N\ell/K\bigr)$ that includes
  precisely the set of codewords $(c_1,\dotsc, c_N)\in C'$ that
  satisfy $c_i\in S_i$ for at least $\alpha N$ values of $i$, provided
\[\alpha > (s+1) (K/N)^{s/(s+1)} \ell^{1/(s+1)}.\]
\end{itemize} 
\end{thm}

\medskip

\noindent {\bf The code}: We set $s=O(\log 1/\epsilon)$, $r=O(1)$, and
$N=K\poly\bigl(\log(1/\epsilon)\bigr)/\epsilon$ in
Theorem~\ref{pv-list} in order to obtain a code $C'\subseteq
\F_{q^s}^N$.  We modify the code, replacing the $i$th coordinate $c_i$
with the pair $(i, c_i)$ for each $i$, in order to obtain a code
$C''$. This latter step only reduces the rate by a constant factor.

Recall that we are trying to recover from a $1/2-\epsilon$ fraction of
deletions. We use Theorem~\ref{binary-list} to construct an inner code
$C_1\colon [N]\times\F_q^s\to \{0,1\}^m$ of rate $\Omega(\epsilon^2)$
which recovers from a $1/2-\delta$ deletion fraction (where we will
set $\delta=\epsilon/4$). Our final code $C$ is a concatenation of
$C''$ with $C_1$, which has rate $\tilde{\Omega}(\epsilon^3)$.

\begin{thm} $C$ is list-decodable from a $1/2-\epsilon$ fraction of deletions in time $N^{\poly(1/\epsilon)}$. 
\end{thm}
\begin{proof} Our algorithm first defines a set of ``decoding windows''. These are intervals of length $(1/2 + \delta)m$ in the received codeword which start at positions $1 + t\delta m$ for $t=0,1,\dotsc, N/\delta -(1/2  + \delta)/\delta$, in addition to one interval consisting of the last $(1/2 +\delta)m$ symbols in the received codeword. 

We use the algorithm of Theorem~\ref{binary-list} to list-decode each decoding window, and let $\mathcal{L}$ be the union of the lists for each window. Finally, we apply the algorithm of Theorem~\ref{pv-list} to $\mathcal{L}$ to obtain a list containing the original message. 
\medskip

{\bf Correctness}: Let $c=(c_1,\dotsc, c_N)$ be the originally transmitted codeword of $C'$. If an inner codeword $C_1(i, c_i)$ has suffered fewer than a $1/2 - 2\delta$ fraction of deletions, then one of the decoding windows is a substring of $C_1(i, c_i)$, and $\mathcal{L}$ will contain the correct pair $(i, c_i)$. 

When $\delta=\epsilon/4$, by a simple averaging argument, we have that an $\epsilon$ fraction of inner codewords have at most $1/2-2\delta$ fraction of positions deleted. For these inner codewords, $\mathcal{L}$ contains a correct decoding of the corresponding symbol of $c$. 
\medskip

In summary, we have list-decoded at most $N/\delta$ windows, with a list size of $O(1/\delta^2)$ each. We also have that an $\epsilon$ fraction of symbols in the outer codeword of $C'$ is correct. Setting $\ell=O(1/\delta^3)$ in the algorithm of Theorem~\ref{pv-list}, we can take $\alpha=\epsilon$. Theorem~\ref{pv-list} then guarantees that the decoder will output a list of $\poly(1/\epsilon)$ codewords, including the correct codeword $c$. 
\end{proof}

\section{Conclusion and open problems}
\label{sec:open}

In this work, we initiated a systematic study of codes for the adversarial deletion model, with an eye towards constructing codes achieving more-or-less the correct trade-offs at the high-noise and high-rate regimes.
There  are still several major gaps in our understanding of deletion codes, and below we highlight some of them 
(focusing only on the worst-case model):
\begin{enumerate}

\item For binary codes, what is the supremum $p^*$ of all
  fractions $p$ of adversarial deletions for which one can have
  positive rate? Clearly $p^* \le 1/2$; could it be that $p^* =
1/2$ and this trivial limit can be matched? Or is it the case that
$p^*$ is strictly less than $1/2$? Note that by~\cite{KMTU}, $p^*>.17$. 

\item The above question, but now for an alphabet of size $k$ --- at what value of $p^*(k)$ does the achievable rate against a fraction $p^*(k)$ of worst-case symbol deletions vanish? It is known that $\frac{1}{k} \le 1-p^*(k) \le O\left(\frac{1}{\sqrt{k}}\right)$ (the upper bound is established in Section~\ref{sec:exis}). Which (if either) bound is asymptotically the right one?

\item Can one construct codes of rate $1-p-\gamma$ to efficiently correct a fraction $p$ of deletions over an alphabet size that only depends on $\gamma$?

Note that this requires a relative distance of $p$, and currently we
only know algebraic-geometric and expander-based codes which achieve
such a tradeoff between rate and relative distance.

\item Can one improve the rate of the binary code construction to correct a fraction $\eps$ of deletions to $1-\eps \poly(\log (1/\eps))$, approaching more closely the existential $1 - O(\eps \log (1/\eps))$ bound?

In the case of errors, an approach using expanders gives the analogous tradeoff (see~\cite{sigact-exp} and references therein). Could such an approach be adapted to the setting of deletions? 

\item Can one improve the $N^{\poly(1/\eps)}$ type dependence of our construction and decoding complexity to, say, $\exp(\poly(1/\eps)) N^c$ for some exponent $c$ that doesn't depend on $\eps$?

\end{enumerate}

\bibliographystyle{abbrv}

\begin{thebibliography}{10}

\bibitem{AEL}
N.~Alon, J.~Edmonds, and M.~Luby.
\newblock Linear time erasure codes with nearly optimal recovery (extended
  abstract).
\newblock In {\em FOCS}, pages 512--519. IEEE Computer Society, 1995.

\bibitem{CMNN}
J.~Chen, M.~Mitzenmacher, C.~Ng, and N.~Varnica.
\newblock Concatenated codes for deletion channels.
\newblock In {\em IEEE International Symposium on Information Theory, 2003},
  pages 218--218, June 2003.

\bibitem{DM}
M.~Davey and D.~J.~C. MacKay.
\newblock Reliable communication over channels with insertions, deletions, and
  substitutions.
\newblock {\em IEEE Transactions on Information Theory}, 47(2):687--698, Feb
  2001.

\bibitem{DG}
S.~Diggavi and M.~Grossglauser.
\newblock On information transmission over a finite buffer channel.
\newblock {\em IEEE Transactions on Information Theory}, 52(3):1226--1237,
  March 2006.

\bibitem{gallager}
R.~Gallager.
\newblock Sequential decoding for binary channels with noise and
  synchronization errors, October 1961.
\newblock Lincoln Lab. Group Report.

\bibitem{sigact-exp}
V.~Guruswami.
\newblock Guest column: error-correcting codes and expander graphs.
\newblock {\em SIGACT News}, 35(3):25--41, 2004.

\bibitem{GI-ieeeit}
V.~Guruswami and P.~Indyk.
\newblock Linear-time encodable/decodable codes with near-optimal rate.
\newblock {\em IEEE Transactions on Information Theory}, 51(10):3393--3400,
  2005.

\bibitem{GR-FRS}
V.~Guruswami and A.~Rudra.
\newblock Explicit codes achieving list decoding capacity: {E}rror-correction
  with optimal redundancy.
\newblock {\em IEEE Transactions on Information Theory}, 54(1):135--150, 2008.

\bibitem{GR-soft}
V.~Guruswami and A.~Rudra.
\newblock Soft decoding, dual {BCH} codes, and better list-decodable e-biased
  codes.
\newblock In {\em Proceedings of the 2008 IEEE 23rd Annual Conference on
  Computational Complexity}, CCC '08, pages 163--174, Washington, DC, USA,
  2008. IEEE Computer Society.

\bibitem{KMS}
A.~Kalai, M.~Mitzenmacher, and M.~Sudan.
\newblock Tight asymptotic bounds for the deletion channel with small deletion
  probabilities.
\newblock In {\em ISIT}, pages 997--1001, 2010.

\bibitem{KM}
Y.~Kanoria and A.~Montanari.
\newblock Optimal coding for the binary deletion channel with small deletion
  probability.
\newblock {\em IEEE Transactions on Information Theory}, 59(10):6192--6219, Oct
  2013.

\bibitem{KMTU}
I.~Kash, M.~Mitzenmacher, J.~Thaler, and J.~Ullman.
\newblock On the zero-error capacity threshold for deletion channels.
\newblock In {\em Information Theory and Applications Workshop (ITA), 2011},
  pages 1--5, Feb 2011.

\bibitem{KLM}
M.~Kiwi, M.~Loebl, and J.~Matou\u{s}ek.
\newblock Expected length of the longest common subsequence for large
  alphabets.
\newblock {\em Advances in Mathematics}, 197:480--498, November 2004.

\bibitem{m-survey}
M.~Mitzenmacher.
\newblock A survey of results for deletion channels and related synchronization
  channels.
\newblock {\em Probability Surveys}, 6:1--33, 2009.

\bibitem{MD06}
M.~Mitzenmacher and E.~Drinea.
\newblock A simple lower bound for the capacity of the deletion channel.
\newblock {\em IEEE Transactions on Information Theory}, 52(10):4657--4660,
  2006.

\bibitem{PV-focs05}
F.~Parvaresh and A.~Vardy.
\newblock Correcting errors beyond the {G}uruswami-{S}udan radius in polynomial
  time.
\newblock In {\em Proceedings of the 46th Annual IEEE Symposium on Foundations
  of Computer Science}, pages 285--294, 2005.

\bibitem{Ratzer}
E.~Ratzer.
\newblock Marker codes for channels with insertions and deletions.
\newblock {\em Annals of Telecommunications}, 60(1--2):29--44, Jan--Feb 2005.

\bibitem{SZ}
L.~Schulman and D.~Zuckerman.
\newblock Asymptotically good codes correcting insertions, deletions, and
  transpositions.
\newblock {\em IEEE Transactions on Information Theory}, 45(7):2552--2557,
  November 1999.

\bibitem{shum-etal}
K.~W. Shum, I.~Aleshnikov, P.~V. Kumar, H.~Stichtenoth, and V.~Deolalikar.
\newblock A low-complexity algorithm for the construction of
  algebraic-geometric codes better than the gilbert-varshamov bound.
\newblock {\em IEEE Transactions on Information Theory}, 47(6):2225--2241,
  2001.

\bibitem{sloane-deletion}
N.~J.~A. Sloane.
\newblock On single-deletion-correcting codes.
\newblock {\em CoRR}, arxiv.org/abs/math/0207197, 2002.

\bibitem{WB}
L.~R. Welch and E.~R. Berlekamp.
\newblock Error correction of algebraic block codes.
\newblock {\em US Patent Number 4,633,470}, December 1986.

\bibitem{zigangirov}
K.~Zigangirov.
\newblock Sequential decoding for a binary channel with drop-outs and
  insertions.
\newblock {\em Problemy Peredachi Informatsii}, 5(2):23--30, 1969.

\end{thebibliography}

\appendix
\section{Omitted proofs}

In this section, we give the omitted probabilistic proofs of Sections~\ref{sec:exis} and~\ref{sec:list}. 

\begin{proof}[Proof of Lemma~\ref{lem:count}]
We will give a way to generate all strings $s'$ containing $s$ as a subsequence, and bound the number of possible outcomes. We do this by considering the lexicographically first occurrence of $s$ in $t$. 
\smallskip

First choose $\ell$ locations $n_1<\dotsb <n_\ell$ in $[m]$, which will be the locations of the $\ell$ symbols of $s$. If the $i$th symbol of $s$ is $a$, we allow all symbols between locations $n_{i-1}$ and $n_i$ to take any value but $a$. This ensures that the locations $n_i$ are the {\em earliest} occurrence of $s$ as a subsequence. The rest of the symbols after $n_\ell$ are filled in arbitrarily. 

It is clear that this process generates any string having $s$ as a subsequence, so we will bound the number of ways this can happen. Fix $n_\ell=t$. There are 
\begin{itemize}
\item $\binom{t-1}{\ell-1}$ ways to choose $n_1,\dotsc, n_{\ell-1}$,
\item $(k-1)^{t-\ell}$ ways to fill in symbols between the $n_i$'s, 
\item and $k^{m-t}$ ways to fill in the last $m-t$ symbols.  
\end{itemize}

Summing over all possible values of $t$, the total number of strings with $s$ as a subsequence is at most 
\[\sum_{t=\ell}^m \binom{t-1}{\ell-1} k^{m-t}(k-1)^{t-\ell}.\]
As $\sum_{t=\ell}^m \binom{t-1}{\ell-1} = \binom{m}{\ell}$, the claimed bound follows. 

When $\ell>m/k$, the term $\binom{t-1}{\ell-1} k^{m-t}(k-1)^{t-\ell}$ increases with $t$, so the sum is at most 
\[\delta m\cdot \binom{m-1}{\ell-1} (k-1)^{m-\ell},\]
giving us our bound for $k=2$. 

\end{proof}

\begin{proof}[Proof of Theorem~\ref{low-del}] We construct such a code using a greedy algorithm. We begin with an arbitrary string in $[k]^m$, and then iteratively add strings whose LCS with all previously chosen strings has length less than $(1-\delta) m$. The LCS of two length $m$ strings can be computed in time $\poly(m)$, so this takes time $k^{O(m)}$. 

It remains to show that we can choose $k^{Rm}$ strings. 
\smallskip

For a fixed string $u\in[k]^m$, it has at most $\binom{m}{(1-\delta)m}$ subsequences of length $(1-\delta) m$, so by Lemma~\ref{lem:count}, the number of strings whose LCS with $u$ has length at least $(1-\delta) m$, and which therefore cannot be chosen, is at most 
\[\binom{m}{(1-\delta) m}^2k^{\delta m}.\]

Thus if the target rate is $R$, we will succeed if 
\[\binom{m}{\delta m}^2 k^{\delta m}\cdot k^{Rm}\leq k^m.\tag{*}\]

It suffices to have

\[2m h(\delta) + \delta m\log k + R m \log k\leq m\log k.\]

Setting $R=1-\delta - \gamma$, we have 
\[2h(\delta) + (1-\gamma)\log k\leq\log k \Leftrightarrow 2 h(\delta) \leq \gamma \log k,\]

so we can choose $k^{Rm}$ strings as long as the alphabet size $k$ satisfies 
\[k\geq2^{2h(\delta)/\gamma}.\]

In the case of $k=2$, we may use the tighter estimate from Lemma~\ref{lem:count} in Equation~(*) to obtain the claimed bound. 
\end{proof}

\begin{proof}[Proof of Proposition~\ref{thm:balanced}] The greedy algorithm of Theorem~\ref{low-del} applies, but now we must choose strings from the set of $\beta$-dense strings. We first bound the number of strings which are \emph{not} $\beta$-dense. The number of strings of length $\beta m$ with less than $\beta m/10$ $1$'s is 
\[\sum_{j=0}^{\beta m/10-1} \binom{\beta m}{j}\leq 2^{h(1/10) \beta m}.\]

Since there are at most $m$ intervals of length $\beta m$ in a string, the probability that a randomly chosen string of length $m$ is not $\beta$-dense is at most 
\[m\cdot \frac{2^{h(1/10)\beta m}}{2^{\beta m}}\leq 2^{-\Omega(\beta m)}.\]

The algorithm of Theorem~\ref{low-del} then succeeds if 
\[\binom{m}{\delta m}^2\cdot \delta m \cdot 2^{Rm}\leq 2^m\bigl(1-2^{-\Omega(\beta m)}\bigr),\]
or $R\leq1-2h(\delta) - O(\log (\delta m)/m) - 2^{-\Omega(\beta m)}/m$. 
\end{proof}

\begin{proof}[Proof of Theorem~\ref{binary-list}]

By Lemma~\ref{lem:count}, the probability that a set of $L$ independent, uniform strings all share a common substring of length $\ell$ is at most 
\[2^\ell \cdot \left(\sum_{t=\ell}^m \binom{t-1}{\ell-1} 2^{-t}\right)^L\leq 2^\ell \left[m^L\cdot 2^{-mL}\cdot \binom{m-1}{\ell-1}^L\right].\]

For a random code $C$ of rate $R$, we union bound over all possible subsets of $L$ codewords to upper bound the probability that $C$ is \emph{not} $(\delta,L)$ list-decodable from deletions. 
\[\Pr[\text{$C$ fails}]< 2^{RmL}\cdot 2^\ell\cdot 2^{L\log m}\cdot 2^{-mL} \cdot 2^{L m h(1-\delta)}.\]

This is at most $2^{-m}$, provided 
\[R\leq 1 - h(\delta) - \frac{2-\delta}{L}-\frac{\log m}{m},\]
which holds for our choice of $R$. 

When $\delta = 1/2 - \epsilon$, we can set $R=\Omega(\epsilon^2)$ to see that 
\[L>\frac{3/2 + \epsilon}{2\epsilon^2/\ln 2 - R - O(\epsilon^3)}\]
so we can take $L$ to be $O(1/\epsilon^2)$. 
\medskip

Similarly to Theorem~\ref{low-del}, this argument shows that we can construct a $\bigl(\delta, O(1/\epsilon^2)\bigr)$ list-decodable code using a greedy algorithm, which successively adds strings who do not share a common subsequence of length $\ell$ with $L-1$ previously chosen strings. 
\end{proof}

\end{document}